%% file: main.tex
\documentclass[letterpaper]{ieeetran}
\usepackage{amsmath,amssymb,amsfonts,amsthm}
\usepackage{xfrac}
\usepackage{mathtools}
\usepackage{float}
\usepackage[font=small]{caption}
\usepackage{subcaption}
\usepackage[dvipsnames]{xcolor}
\usepackage[capitalize]{cleveref} % cool stuff for cross-reference
\usepackage{makecell}
\let\labelindent\relax
\usepackage{enumitem}  
\usepackage{algorithm}
\usepackage{algpseudocode}
\usepackage{pgfplots}
\usepackage{tikz-network}
\DeclareUnicodeCharacter{2212}{−}
\usepgfplotslibrary{groupplots,dateplot}
\usetikzlibrary{patterns,shapes.arrows}
\usetikzlibrary{pgfplots.statistics}
\pgfplotsset{compat=newest}

\usepackage{pifont}

\DeclareMathAlphabet{\mymathbb}{U}{BOONDOX-ds}{m}{n}

\newcounter{thms}
\newtheorem{theorem}[thms]{Theorem}
\newtheorem{corollary}[thms]{Corollary}

\newcounter{asses}
\newtheorem{assumption}[asses]{Assumption}
\newcounter{defs}
\newtheorem{definition}[defs]{Definition}
\newcounter{rems}
\newtheorem{remark}[rems]{Remark}
\theoremstyle{remark}

\newcommand{\R}{\mathbb{R}}

\newcommand{\bb}[1]{\mymathbb{#1}}

\newcommand{\bs}[1]{\boldsymbol{#1}}
\newcommand{\mc}[1]{\mathcal{#1}}
\newcommand{\mt}[1]{\textnormal{#1}}

\newcommand{\lt}{\left}
\newcommand{\rt}{\right}
\newcommand{\beeq}{\begin{equation}}
\newcommand{\eneq}{\end{equation}}
\newcommand{\matb}{\begin{matrix}}
\newcommand{\mate}{\end{matrix}}

\DeclareMathOperator*{\argmin}{arg\,min}

\definecolor{redtmp}{rgb}{0, 0, 0}
\definecolor{redtmp2}{rgb}{0, 0, 0}
\colorlet{myr}{Maroon}
\colorlet{myg}{OliveGreen}
\colorlet{myb}{NavyBlue}
\colorlet{myy}{Dandelion}

\title{\LARGE \bf
Regularization for distributionally robust state estimation and prediction
}

\author{Jean-Sébastien Brouillon, Florian Dörfler, and Giancarlo Ferrari-Trecate% <-this % stops a space
\thanks{This research is supported by the Swiss National Science Foundation under the NCCR Automation (grant agreement 51NF40\_180545).}% <-this % stops a space
\thanks{J.S. Brouillon and G. Ferrari-Trecate are with the Institute of Mechanical Engineering, \'Ecole Polytechnique F\'ed\'erale de Lausanne, Switzerland. Email addresses: {\tt\small \{jean-sebastien.brouillon, giancarlo.ferraritrecate\}@epfl.ch}.}%
\thanks{F. Dörfler is with the Automatic Control Laboratory, Swiss Federal Institute of Technology (ETH), Switzerland. Email address: {\tt\small dorfler@control.ee.ethz.ch}.}%
}

\begin{document}
%\bstctlcite{IEEEexample:BSTcontrol}

\maketitle
\thispagestyle{empty}
%\pagestyle{empty}

%%%%%%%%%%%%%%%%%%%%%%%%%%%%%%%%%%%%%%%%%%%%%%%%%%%%%%%%%%%%%%%%%%%%%%%%%%%%%%%%
\begin{abstract}
The increasing availability of sensing techniques provides a great opportunity for engineers to design state estimation methods, which are optimal for the system under observation and the observed noise patterns. However, these patterns often do not fulfill the assumptions of existing approaches. We provide a direct method using samples of the noise to create a moving horizon observer for linear time-varying and nonlinear systems, which is optimal under the empirical noise distribution. Moreover, we show how to enhance the observer with distributional robustness properties in order to handle unmodeled components in the noise profile, as well as different noise realizations. We prove that, even though the design of distributionally robust estimators is a complex minmax problem over an infinite-dimensional space, it can be transformed into a regularized linear program using a system level synthesis approach. Numerical experiments with the Van der Pol oscillator show the benefits of not only using empirical samples of the noise to design the state estimator, but also of adding distributional robustness. We show that our method can significantly outperform state-of-the-art approaches under challenging noise distributions, including multi-modal and deterministic components.
\end{abstract}


%%%%%%%%%%%%%%%%%%%%%%%%%%%%%%%%%%%%%%%%%%%%%%%%%%%%%%%%%%%%%%%%%%%%%%%%%%%%%%%%
\section{Introduction} \label{section_intro}
\IEEEPARstart{E}{stimating} and predicting the states of a system is a fundamental problem in many areas of science and engineering, ranging from control theory to signal processing and machine learning. The goal is to use a set of noisy and possibly incomplete observations of the system's output to infer the true internal state of the system with minimal error. The problem of state smoothing\textcolor{redtmp}{, filtering,} and prediction (hereafter referred to as state estimation problem, for short) is challenging due to several factors, such as the presence of measurement noise, unmodeled dynamics, nonlinearities, and uncertainty. The recent advances in sensing and communications technologies and computation have allowed engineers to gather large amounts of data about the noise affecting systems of various nature.

The design of a high-performance state estimator for a given system follows three steps: (i) the accurate modeling of the system dynamics and the statistics of the process and measurement noises, (ii) the choice of an estimator that best fits the model and noise assumptions, and (iii) the optimization of the estimator parameters. This process can be difficult, especially if the noises follow an uncommon profile (e.g., including outliers or deterministic signals), or if the system is time-varying. In the latter case, the design process must be repeated online.% after each change in the dynamics, or offline for every possible situation.

%(i) the accurate modeling of the system dynamics and the statistics of the process and measurement noises, (ii) the choice of an estimator that best fits the model and noise assumptions, and (iii) the optimization of the estimator parameters.
%First, both the system dynamics and the statistics of process and measurement noises must be accurately modeled. Second, one chooses the estimation method that fits the problem best. Third, the required parameters to implement this method must be computed from the system and noise models.

The most popular estimation method is the Kalman Filter (KF), which has a closed form solution that can be computed online. This is the backbone of the Extended Kalman Filter (EKF), which recomputes the filter parameters at each time step based on the linearization of a system at the current operating point \cite{ekf_og}. KFs may not perform well when the variance is not accurately measured, even if the noise is Gaussian. To address this issue, \cite{colored_noise_gp} proposes an automatic method for learning KF parameters. Another popular estimation method is to stabilize the error dynamics and reject errors in the initial state estimate using a Luenberger Observer (LO). While the KF provides optimality guarantees for linear systems under Gaussian noise, the LO can be a better candidate for other noise distributions, even though its optimal design is challenging in real time.

%The most common method for state estimation is the Kalman filter, which gives a minimal error under the assumption that the process is linear and both the process noise (or disturbance) and the measurement noise are Gaussian \cite{kalman_og}. If no Gaussian distribution is observable in the noise samples, a popular method is to stabilize the error dynamics to eventually reject errors in the initial state estimate using a Luenberger Observer (LO) \cite{luenberger}. The LO is simpler to implement and does not rely on a specific noise distribution, but its performance and robustness are not optimal. % To handle modeling errors, newer and improved versions of both KFs and LOs have been developed, e.g., the Extended Kalman filter \cite{ekf_og} for nonlinearities. However, these methods are often designed with problem-specific assumptions in mind.

When dealing with non-Gaussian disturbances, particle filters are a popular approach, but they are computationally expensive and do not exploit specific patterns in non-stochastic noise profiles. Other methods involve learning the non-stochastic part of the noise and assuming standard Gaussian or worst-case distributions for the stochastic component \cite{nonstoch_gauss, nonstoch_bounded}. However, these approaches still make strong assumptions about the noise, which can lead to poor performance if they are not verified. A more flexible method is Moving Horizon Estimation (MHE). It can model not only non-stochastic profiles by penalizing combinations of errors at different time steps, but also non-Gaussian noise distributions using non-quadratic cost functions\footnote{Although we focus on the unconstrained case in this paper, MHE can also implement constraints \cite{alessandri2008moving}.} \cite{book_mpc}. However, MHE requires significant computing power and can be sensitive to modeling errors in both the noise statistics and the system itself \cite{robust_mhe}.

Distributionally Robust Optimization (DRO) is a powerful mathematical tool to mitigate errors in the statistical modeling of the noise, by considering the worst probability distribution within an uncertainty set around the empirical one \cite{wasserstein_kuhn}. Recent advances in this field have significantly simplified the computation of robust optimizers, by showing the equivalence between distributional robustness and regularization \cite{soroosh_mass_transport, multivar_relax_dro_chen}. DRO has recently been applied to Model Predictive Control (MPC) and Data-enabled Predictive Control (DeePC) to provide a direct method from noise samples to controller design \cite{dr_mpc, dr_mpc_liviu, dr_deepc}. This approach has only been applied to the field of state estimation under the assumption that the worst case distribution is Gaussian \cite{dr_mhe, dr_mhe2, dr_ekf}.

%This paper provides a novel robust MHE method based on DRO, which directly incorporates samples of the noise profile, eliminating any statistical modeling work. Moreover, because the DRO literature only provides solutions for a relaxed problem when the parameters are two-dimensional, we provide a new proof that regularization is exactly equivalent to distributional robustness for a common type of uncertainty sets. Finally, we show that the final estimation problem is a combination of several small Linear Programs (LPs), which are very fast to solve.

In this paper, we attempt to fill the gap and introduce a robust unconstrained MHE method that uses DRO to incorporate samples of the noise profile directly in the estimation process, hence eliminating the need for statistical modeling. To do so, \textcolor{redtmp2}{we prove that the regularization-based relaxation proposed in \cite{multivar_relax_dro_chen} can be exact for $\ell_1$ norm-based loss functions, which are relevant for MHE. This extends the results obtained for vector-valued parameters in \cite{wasserstein_kuhn, soroosh_mass_transport}}. 
%we provide a new proof that regularization is equivalent to distributional robustness for matrix-valued parameters and common uncertainty, which extends the results for vector-valued parameters in \cite{wasserstein_kuhn, soroosh_mass_transport} and the relaxations in \cite{multivar_relax_dro_chen}. 
We show that our approach is capable of providing both predictions and filtered state estimates for discrete-time linear time-varying systems. % and nonlinear systems that are linearized at each time step around the current operation point (similar to the EKF). 
Moreover, the final estimation problem is a combination of several small Linear Programs (LPs), which can be efficiently solved in real time. Finally, we provide a simulation example illustrating the performance of this new method for the observation of a linearized Van der Pol oscillator under challenging noise profiles. %in the context of a pacemaker tracking the state of the heart modelled by a Van der Pol oscillator \cite{vdp_pacemaker}.

\subsection{Preliminaries and Notation}
\label{subsec:notation}
Time indices are denoted by the subscript $t$, and boldface letters denote the stacked vectors at all times in a window. Similarly, calligraphic letters denote linear operators applying to such stacked vectors. Underlined bold symbols are trajectory matrices, whose columns are bold symbol vectors. For example, for a state $x_t \in \bb R^n$, the trajectory over the window $[t-T,t]$ is $\bs x = [x_{t-T}^\top, \dots, x_t^\top]^\top \in \bb R^{n(T+1)}$, which can be affected by the operator $\mc C$ such that $\bs y = \mc C \bs x$. If $N$ of these trajectories are available, they can be included in the matrix $\underline{\bs x} = [\bs x_0, \dots, \bs x_N] \in \bb R^{n(T+1) \times N}$.% In the whole paper, we call noise both disturbance (or process noise) and measurement noise, and specify its type when needed.

The subscript $i$ is used to denote the $i^{th}$ row of a matrix. The matrix $I$ denotes the identity and $I_i$ is the $i^{th}$ unit vector. The function $\mt{blkdiag}([X_0, \dots, X_N])$ constructs a block-diagonal matrix from the blocks $X_0, \dots, X_N$. %The matrix $\mc Z$ is the block-downshift operator and is constructed by shifting the diagonal of the identity matrix down $n$ times for blocks of size $n$.

The $\ell_2$-norm of a vector is denoted by $\|\cdot\|_2$, which also denotes the spectral norm of a matrix (its largest singular value). The $\ell_1$ norm of a vector is denoted by $\|\cdot\|_1$, and $\|\cdot\|_{F_1}$ is the $\ell_1$ Frobenius norm of a matrix, i.e. the sum of the $\ell_1$ norms of its rows.

\section{System model} \label{section_model}
%\subsection{Estimation} \label{subsec_estim_model}

%We model a generic system in the state-space with Linear Time-Varying (LTV) dynamics, subject to unknown disturbance and measurement noise. In order to estimate and predict the system's states, we design a sequence of Luenberger observers, which apply to several time steps in a window around the present. %This policy-oriented estimator allows to mitigate the error over a set of noise realizations rather than only the current one, as well as to react to time correlations in the same way as MHE.

\subsection{LTV dynamics and observer}\label{subsec_dyn}

We model a dynamical system using a discrete-time state-space representation, where the state $x_t \in \bb R^n$ is hidden, and only the output $y_t \in \bb R^p$ is observed. The state dynamics and output map are fully described by the equations
\begingroup
\abovedisplayskip=4pt
\belowdisplayskip=4pt
\begin{subequations}\label{eq_model_system}
\begin{align}
    x_{t+1} &= A_t x_t + w_{t}, \\
    y_{t} &= C_t x_t + v_{t},
\end{align}
\end{subequations}
\endgroup
where $w_t$ and $v_t$ are generic process and measurement noises.
\begin{assumption}\label{ass_obs}
    The system \eqref{eq_model_system} is observable for all $t$.
\end{assumption}
\noindent
This assumption in very common and often necessary to estimate the states of a system \cite{luenberger}.

We aim to compute the estimates $\hat x_\tau$ of the states in the window $[t-T_s, \dots, t+T_f]$ around the current time $t$. To do so, we use the following state estimator
\begingroup
\abovedisplayskip=5pt
\belowdisplayskip=3pt
\vspace{0pt}
\begin{align}\label{eq_observer_apply}
\nonumber \\[-12pt]
    \hat x_{\tau+1} &= A_{\tau} \hat x_{\tau} - \sum_{k=t-T_s}^{t} L_{\tau,k} (C_{k} \hat x_{k} - y_{k}),
\end{align}
\endgroup
which uses the observations $y_k$ for $k = t-T_s, \dots, t$, and design the gains $L_{\tau,k}$ for $\tau = t-T_s, \dots, t+T_f-1$. %This gives the state estimate as follows.
The observer gains must stabilize the dynamics of the error $e_\tau = \hat x_\tau - x_\tau$, given by
\begingroup
\belowdisplayskip=3pt
\vspace{-6pt}
\begin{align}\label{eq_observer_apply_error}
    \textcolor{redtmp}{e_{\tau+1} = A_{\tau} e_{\tau} - w_\tau - \sum_{k=t-T_s}^{t} L_{\tau,k} (C_{k} e_{k} - v_k).}
\end{align}
\endgroup

There are two main differences between \eqref{eq_observer_apply} and the classical MHE problem \cite{mhe_og_rao}: (i) the presence of a forecasting horizon $[t+1, t+T_f]$ after the standard smoothing horizon $[t-T_s, t]$ and (ii) the optimization variables are matrix gains $L_{\tau,k}$, rather than the point estimates $\hat x_{\tau}$. \textcolor{redtmp}{This policy-based problem, similar to dynamic programming for control \cite[Chapter 3.3]{book_mpc}, improves the estimate's robustness, while giving the same results as classic MHE in nominal conditions.}  %This allows one to both maximize both the accuracy of the final estimate and minimize the effect of the noise \cite{book_mpc}. %In contrast, point estimates are faster to compute but requires further calculations to study the impact of the noise (e.g., through a sensitivity matrix).
%\cite{book_mpc}

\begin{remark}
Known system inputs are not included in the observer design since they cancel out when computing the error $e_\tau = \hat x_\tau - x_\tau$. If present, they can be added to \eqref{eq_observer_apply} when computing the state estimate.
\end{remark}

\begin{remark}\label{rem_nonlin}
There are no assumptions on both $v_t$ and $w_t$, which can also include modelling errors. For example, if \eqref{eq_model_system} representes the linearization of the system $x_{t+1} = f(x_t, t) + \tilde w_{t}$, $y_{t} = h(x_t, t) + \tilde v_{t}$ around a state trajectory% by setting $A_t = \nabla_x f(\hat x_t, t)$, $C_t = \nabla_x h(\hat x_t,t)$
, the variables $v_t$ and $w_t$ can embed \textcolor{redtmp}{worst-case} linearization errors.
% \begin{align}\label{eq_observer_apply_nonlin}
%     \hat x_{t'+1} &= f(\hat x_{t'}, t') - \sum_{\tau=t-T_s}^{t} L_{t',\tau} (h(\hat x_{\tau}, \tau) - y_{\tau}).
% \end{align}
\end{remark}
%In Section \ref{section_results}, we take advantage of the aforementioned linearization to study a nonlinear system. 
In the sequel, we consider the estimation problem for a single horizon with a fixed $t$. Hence, for simplicity $t$ is omitted in the notation.

\subsection{Error dynamics over the entire horizon}\label{subsec_mats_def}
To design the LTV observer policy based on the gains $L_{\tau,k}$ for $\tau \in [t-T_s, t+T_f], k \in [t-T_s, t]$, we stack the dynamics of the state estimation error \eqref{eq_observer_apply_error} as
\begin{align}\label{eq_error_sys_aug}
    \bs e = \mc Z \mc A \bs e - \mc L \mc C \mc Z \bs e + \mc L \bs v + \bs w,
\end{align}
\vspace{-12pt}
\begingroup % keep the change local
\renewcommand{\arraystretch}{0.85}
\begin{subequations}\label{eq_def_FHvw}
\begin{align}\label{eq_def_mcac}
&\mt{where } 
\mc Z = \!\!\lt[\matb \vspace{-8pt}\\
    0_{n\times n} & & & \\
    \vspace{-16pt}
    \\ \hspace{-4pt}I & \hspace{-12pt}\ddots & & \\
    \vspace{-16pt}
    \\ & \hspace{-16pt}\ddots & \hspace{-8pt}\ddots &
    \\ & & \hspace{-16pt}I & \hspace{-8pt}0_{n\times n}
\mate \rt]\!\!,\,
\\
&\mc A = \!\!\lt[\matb \vspace{-6pt}\\
    A_{t-T_s}\\
    \vspace{-15pt}
    \\ & \hspace{-20pt}\ddots\hspace{-20pt} & &
    \\ & & \hspace{-20pt} A_{t+T_f-1} &
    \\ & & & \hspace{-20pt}0_{n\times n}
\mate \rt]\!\!,\,
\mc C = \!\!\lt[\matb \vspace{-6pt}\\
    0_{n\times p} & & & 
    \\ & \hspace{-16pt} C_{t-T_s} & & \\
    \vspace{-15pt}
    \\ & & \hspace{-20pt}\ddots\hspace{-20pt} &
    \\ & & & \hspace{-20pt}C_{t+T_f-1}
\mate \! \rt]\!\!,
\\
&\,\bs v = \!\!\lt[\matb 0_{p \times 1} \\ v_{t-T_s} \\ \vspace{-15pt}\\ \vdots \\ v_{t+T_f-1} \mate\rt]\!\!,\,
\bs w =\! \! \lt[\matb e_{t-T_s} \\ -w_{t-T_s} \\ \vspace{-15pt}\\ \vdots \\ -w_{t+T_f-1} \mate\rt]\!\!,\,
\bs e = \!\!\lt[\matb e_{t-T_s} \\ \vspace{-15pt}\\ \vdots \\ e_{t+T_f-1} \\ e_{t+T_f} \mate\rt]\!\!,
\end{align}
\setlength{\arraycolsep}{2pt}
% \begin{subequations}\label{eq_def_FHvw}
% \begin{align}\label{eq_def_phie}
% \mc A &= \lt[\matb 
%     \nabla_x f(\hat x_{t-T_s},t-T_s) \hspace{-100pt}
%     \\ & \hspace{-50pt} \ddots & &
%     \\ & & \hspace{-50pt} \nabla_x f(\hat x_{t+T_f-1}, t+T_f-1) \hspace{-50pt} &
%     \\ & & & \hspace{-70pt} 0_{n\times n}
% \mate \hspace{-9pt} \rt],
% \\
% \mc C &= \lt[\hspace{-5pt}\matb 
%     \hspace{25pt} 0_{n\times p} & & & 
%     \\ & \hspace{-50pt} \nabla_x h(\hat x_{t-T_s+1},t-T_s+1) \hspace{-50pt} & & 
%     \\ & & \hspace{-50pt} \ddots &
%     \\ & & & \hspace{-50pt} \nabla_x f(\hat x_{t+T_f}, t+T_f)
% \mate\rt],
% \\
% \bs v &= \!\lt[\matb 0 \cdot v_{t} \\ v_{t-T_s} \\ \vdots \\ v_{t+T_f-1} \mate\rt]\!\!,
% \bs w =\! \! \lt[\matb e_{t-T_s} \\ -w_{t-T_s} \\ \vdots \\ -w_{t+T_f-1} \mate\rt]\!\!,
% \bs e = \lt[\matb e_{t-T_s} \\ \vdots \\ e_{t+T_f-1} \\ e_{t+T_f} \mate\rt]\!\!,
% \end{align}
% \end{subequations}
and $\mc L$ is the observer policy to be designed, written as
\begin{align}\label{eq_L_def}
    \!\!\mc L = \!\!\lt[\matb \vspace{-6pt}\\
    0_{n \times p}\! & 0_{n \times p} & \dots & 0_{n \times p} & \!0_{n \times p(T_f-1)}
    \\
    0_{n \times p}\! & L_{t-T_s, t-T_s} & \dots & L_{t-T_s, t} & \!0_{n \times p(T_f-1)}
    \\
    \vdots & \vdots & \vdots & \vdots & \vdots
    \\
    0_{n \times p}\! & L_{t+T_f-1, t-T_s} & \dots & L_{t+T_f-1, t} & \!0_{n \times p(T_f-1)}
    \mate\rt]\!\!\!. \!\!
\end{align}
\end{subequations}
\endgroup
The last block-columns in $\mc L$ are zero to ensure causality, meaning that the last $p(T_f-1)$ measurements in the window $[t-T_s, t+T_f]$, which are in the future, can not be used. The zero first block-column and -row allow one to ensure the equivalence between \eqref{eq_observer_apply} and \eqref{eq_error_sys_aug}, as the first $n$ equations of \eqref{eq_error_sys_aug} only ensure that the initial error $e_{t-T_s}$ is correctly propagated over time.

Note that $\mc A$ and $\mc C$ are block-diagonal matrices and that $\bs v$ and $\bs w$ include both the noise and the modelling errors. Moreover, the error on the initial state is embedded in the first block of the disturbance vector, i.e. $[\bs w_{1}, \dots, \bs w_{n}]^\top = x_{t-T_s} - \hat x_{t-T_s}$. In the sequel, we note the dimensions of $\bs v$ and $\bs w$ as $\bs p = p(T_s + T_f +1)$ and $\bs n = n(T_s + T_f +1)$, respectively.

\vspace{-2pt}
\section{Problem statement} \label{subsec_problem}

%Classical MHE problems aim to find the estimates $\hat x_t$, which minimize a nominal cost function $\mt{cost}(e_t)$ of the estimation error. This nominal approach works well under specific noise distributions, but its performance can be seriously reduced by unforeseen noise realizations. To mitigate the impact of such noise, we propose a distributionally robust approach, which minimizes the worst-case empirical risk generated by the estimation error. 
We aim to design an optimal data-driven observer $\mc L$ from $N$ noise samples $\tilde{\bs v}_i$ and $\tilde{\bs w}_i$ for $i = 1, \dots, N$ collected offline, \textcolor{redtmp}{e.g. during tests prior to the deployment of the observer in the field.} A naive approach is to \textcolor{redtmp}{maximize the likelihood based on the} empirical distributions $\tilde{\bb P}_v(\bs v) = \frac{1}{N}\sum_{i=1}^N \delta(\bs v - \tilde{\bs v}_i)$ and $\tilde{\bb P}_w(\bs w) = \frac{1}{N}\sum_{i=1}^N \delta(\bs w - \tilde{\bs w}_i)$, where $\delta(\cdot)$ is the Dirac distribution. This nominal method gives minimal errors for realizations of the noise that were in the training set, but can lead to a brittle estimator with poor out-of-sample performance. We introduce distributional robustness with respect to the worst case empirical risk in order to mitigate the impact of unforeseen noise realizations.

The worst-case empirical risk is given by the expected cost given by the worst possible probability distribution. For probability distributions in the sets $\bb V$ and $\bb W$ (i.e., $v \sim \bb P_v \in \bb V$ and $w \sim \bb P_w \in \bb W$), the worst-case empirical risk is defined by
\begin{align}\label{eq_risk_def}
\textcolor{redtmp}{
    \mc R(\bs e(\mc L, \bs v, \bs w)) := \sup_{\substack{
    \bb P_v \in \bb V
    \\
    \bb P_w \in \bb W
    %\\
    %\mc A \mt{ s.t. } \|\mc A - \mc F\|_F \leq \varepsilon_A
    %\\
    %\mc C \mt{ s.t. } \|\mc C - \mc U\|_F \leq \varepsilon_A
    }} 
    \bb E_{\substack{v \sim \bb P_v \\ w \sim \bb P_w}} 
    \; \mt{cost}(\bs e(\mc L, \bs v, \bs w)).
}
\end{align}
\begin{assumption}\label{ass_cost}
The estimation cost is $\textnormal{cost}(\bs e) = \|\mc Q \bs e\|_1$, \textcolor{redtmp}{where $\mc Q \in \bb R^{\bs n \times \bs n}$.}
\end{assumption}
\noindent
Although quadratic costs are more common in engineering applications, $\ell_1$ costs are often used for their robustness to non-Gaussian noise \cite{lad_regression}. 
%The main limitation of $\ell_1$ Wassterstein balls is that the cost must be Lipschitz \cite{soroosh_mass_transport}. Otherwise, the expected cost given by a probability distribution can grow faster than the Wassertein norm between this distribution and the center of the ball, which can lead to an infinite risk. To avoid this situation, we introduce the following assumption.

The sets $\bb V$ and $\bb W$ are infinite-dimensional. In order to define them in a meaningful way, we introduce the following definition and assumption.
\begin{definition}\label{def_wasser}
The Wasserstein metric $W_1$ based on the $\ell_\infty$ norm is defined as
\begin{align*}
    W_1(\bb P_1, \bb P_2) = \inf_{\Pi} \int_{\Xi^2} \|\xi_1 - \xi_2\|_\infty \Pi(d\xi_1, d\xi_2),
\end{align*}
where $\Xi$ is the support of $\bb P_1$ and $\bb P_2$ and $\Pi$ is a joint distribution of $\xi_1$ and $\xi_2$ with marginal distributions $\bb P_1$ and $\bb P_2$, respectively.
\end{definition}
\begin{assumption}\label{ass_proba}
The sets $\bb V$ and $\bb W$ are Wasserstein-1 balls $\bb B_{\varepsilon_v}(\tilde{\bb P}_v)$ and $\bb B_{\varepsilon_w}(\tilde{\bb P}_w)$ given by
\begin{subequations}%\label{eq_balls_def}
\begin{align*}
    \bb V &= \bb B_{\varepsilon_v}(\tilde{\bb P}_v) = \{\bb P_v | W_1(\bb P_v, \tilde{\bb P}_v) \leq \varepsilon_v\},
    \\
    \bb W &= \bb B_{\varepsilon_w}(\tilde{\bb P}_w) = \{\bb P_w | W_1(\bb P_w, \tilde{\bb P}_w) \leq \varepsilon_w\}.
\end{align*}
\end{subequations}
The support $\Xi$ of the Wasserstein metric $W_1$ is the entire space $\bb R^{\bs p}$ or $\bb R^{\bs n}$ for $\bb P_v$ and $\bb P_w$, respectively.
\end{assumption}
%\vspace{-20pt}
In Definition \ref{def_wasser}, Wasserstein-1 balls\footnote{\textcolor{redtmp}{The first order Wasserstein metric is the most common in the literature because it is one of the easiest to interpret and reformulate in a tractable way \cite{wasserstein_kuhn, soroosh_mass_transport}.}} only require a norm, a center, and a radius to define a set in the infinite-dimensional space of probability distributions. We chose the $\ell_\infty$ norm because it treats each entry of the noise vectors separately\footnote{\textcolor{redtmp}{Other norms could be less conservative but the problem must be relaxed as in \cite[Theorem 2.1]{multivar_relax_dro_chen} to become tractable.}}, making it easier for a user to determine the radii $\varepsilon_v$ and $\varepsilon_w$. These radii are given by the expected amount of noise in the worst sensor and the expected disturbance in the most perturbed state. The center is a distribution, which is a function with an unbounded support, and thus much harder to determine. We do away with this difficulty by centering the balls on empirical distributions.

\begin{remark}
Robustness against worst-case bounded noise is more common and avoids infinite-dimensional problems. However, this often means one must choose between over-conservatism or lack of robustness if rare noise realizations are very large. Distributional robustness allows one to consider unbounded disturbances, while weighting them in accordance with their probability. This approach is therefore better suited to generic disturbance patterns.
\end{remark}

%Assumption \ref{ass_proba} is very standard in the DRO literature \cite{wasserstein_kuhn} It may seem restrictive because it resembles the use of spherical uncertainty sets in robust MHE. However, this analogy does not apply because the balls contain distributions, which can result in realizations that are infinitely far away from the samples $\tilde{\bs v}_i$ and $\tilde{\bs w}_i$. This happens because the the $\ell_1$ Wasserstein norm is a mass transportation problem where mass refers to the area under a probability density curve. By moving a small amount of this mass while keeping the ball radius constant, it is possible to obtain a non-zero probability for realizations that are far away from the empirical distribution located at the center of the ball.

% \begin{figure}[H]
%     \centering
%     \input{}
%     \caption{One-dimensional example of distributions within an $\ell_1$ Wassterstein ball of radius 0.5 around the Dirac delta distribution.}
%     \label{fig_wassterstein}
% \end{figure}

Under the assumptions \ref{ass_proba} and \ref{ass_cost}, the worst-case empirical risk \eqref{eq_risk_def} becomes
\begin{subequations} \label{eq_prob_def}
\begin{align}\label{eq_cost_def}\textcolor{redtmp}{
    \mc R(\bs e(\mc L, \bs v, \bs w)) =\!\!\sup_{\substack{
    \bb P_v \in \bb B_{\varepsilon_v}(\tilde{\bb P}_v)
    \\
    \bb P_w \in \bb B_{\varepsilon_w}(\tilde{\bb P}_w)
    %\\
    %\mc A \mt{ s.t. } \|\mc A - \mc F\|_F \leq \varepsilon_A
    %\\
    %\mc C \mt{ s.t. } \|\mc C - \mc U\|_F \leq \varepsilon_A
    }}
    \!\bb E_{\substack{v \sim \bb P_v \\ w \sim \bb P_w}} 
    \;
    \|\mc Q \bs e(\mc L, \bs v, \bs w)\|_1,}
\end{align}
and an optimal policy can be computed as
\begin{align}\label{eq_prob_def_L}
    \textcolor{redtmp}{\mc L^\star = \underset{\mc L \mt{ causal}}{\mt{arg}\, \mt{inf}} \; \inf_{\bs e} \mc R(\bs e(\mc L, \bs v, \bs w)) \quad\quad \mt{s.t. \eqref{eq_error_sys_aug}},}
\end{align}
\end{subequations}
\textcolor{redtmp}{which is coupled to \eqref{eq_cost_def} through the constraint \eqref{eq_error_sys_aug}, and} where the constraint "$\mc L \mt{ causal}$" enforces the sparsity pattern given by the zero blocks in \eqref{eq_L_def}. %Note that, to obtain a compact notation, we include the setting of the first block-column and -row to zero in the causality constraint.

% We include the first block-column and -row in the notion of causality even if its

% TODO: In theory just last p0 are zero, but first block-column and -row are there for previous error propagation (consistency with (2))

%However, since the first block of both $\mc C$ and $\bs v$ are zero, the risk $\mc R(\bs e)$ is equal for any values of the first $n$ columns of $\mc L$. We can therefore relax the constraint and define a causal $\mc L$ as one where the $\bs p_0 = p(T_f-1)$ are zero.

% find $\mc L$ minimizing:
% \begin{align}\label{eq_prob_def}
%     &\argmin_{\mc L \mt{ causal}} \mc R(\bs e(\bs v, \bs w, \mc L))
%     \\ \label{eq_prob_dyn_def}
%     &\mt{s.t.}\; \bs e = \mc Z \mc A \bs e - \mc L \mc C \mc Z \bs e + \mc L \bs v + \bs w,
% \end{align}
% where $\mc R(\bs e, \varepsilon)$ is the empirical risk. 

% \vspace{5cm}

\section{Tractable reformulation}
\label{section_tract}
At first glance, the problem \eqref{eq_prob_def} seems very challenging to solve. It is a non-convex, infinite-dimensional, and inf-sup problem. In this section, we first address the non-convexity, and then provide a closed-form solution for the risk $\mc R(\bs e)$. In the end, the problem \eqref{eq_prob_def} is reduced to a simple LP.

\subsection{Convexification}
\label{section_sls}
The first challenge is addressed by proposing a convex reformulation using a System Level Synthesis (SLS) representation of the estimation problem \cite{sls_og}\textcolor{redtmp}{, which decouples \eqref{eq_cost_def} and \eqref{eq_prob_def_L}}. To do so, we \textcolor{redtmp}{use} the disturbance-to-error and noise-to-error maps $\Phi_w \in \bb R^{\bs n \times \bs n}$ and $\Phi_v \in \bb R^{\bs n \times \bs p}$ \textcolor{redtmp}{defined in} \cite{ifac_paper}, i.e.,
\begin{subequations}%\label{eq_def_phi}
\begin{align*}%\label{eq_def_phie}
    \Phi_w &= (I - \mc Z(\mc A - \mc L \mc C))^{-1}, \\
    \Phi_v &= (I - \mc Z(\mc A - \mc L \mc C))^{-1} \mc L.
    %\label{eq_def_phil}
\end{align*}
\end{subequations}
The estimation error is given by $\bs e = \Phi_v \bs v + \Phi_w \bs w$ and the risk \eqref{eq_cost_def} can be rewritten as a function of $\Phi_v$ and $\Phi_w$ as

\begin{align}\label{eq_cost_def_sls}
    \!\!\mc R(\Phi_v, \Phi_w) 
    \!=\!\!
    \sup_{\substack{
    \bb P_v \in \bb B_{\varepsilon_v}(\tilde{\bb P}_v)
    \\
    \bb P_w \in \bb B_{\varepsilon_w}(\tilde{\bb P}_w)}}
    \!\!\!
    \bb E_{\substack{v \sim \bb P_v \\ w \sim \bb P_w}} 
    \|\mc Q(\Phi_{v,i} \bs v + \Phi_{w,i} \bs w)\|_1.\!\!
\end{align}
Using \eqref{eq_cost_def_sls}, the problem \eqref{eq_prob_def} becomes
\begin{align}\label{eq_prob_def_sls}
    &\underset{\substack{\Phi_v \mt{ causal} \\ \Phi_w \mt{ causal}}}{{\mt{arg}\, \textcolor{redtmp}{\mt{inf}}}} \mc R(\Phi_v, \Phi_w)
    \\ \label{eq_prob_achiev_def}
    &\mt{s.t. } [\Phi_v, \Phi_w] \lt[\matb \mc C \mc Z \\ I - \mc Z \mc A \mate \rt] = I.
\end{align}
According to \cite{sls_og}, \eqref{eq_prob_def_sls} is convex in $\Phi_v$ and $\Phi_w$ if and only if $\mc R$ is convex. Moreover, because \eqref{eq_prob_achiev_def} must be satisfied, there exist a $\mc L = \Phi_w^{-1}\Phi_v$ solving \eqref{eq_prob_def}\footnote{\textcolor{redtmp}{By contradiction, assume that $\Phi_w^{-1} = (I-\mc Z(\mc A - \mc L \mc C))$ is not invertible. Under Assumption 3, the Wasserstein ball $\bb B_{\varepsilon_w}$ is supported by $\bb R^{\bs n}$. Hence, $\bb V$ and $\bb W$ always contain realizations of the noises satisfying $\bs{w} + \mc L \boldsymbol{v} \notin \text{span}(I-\mc Z(\mc A - \mc L \mc C)) \subset \bb R^{\bs n}$, which would invalidate the dynamics (4). This contradiction proves that $\Phi_w$ must have maximal rank.}}. The causality constraints here need to result in a causal policy $\mc L$ (see Section \ref{subsec_mats_def}). On the one hand, the noise-to-error map $\Phi_v$ must have the same zero columns as $\mc L$ \textcolor{redtmp}{given in \eqref{eq_L_def}}, because this sparsity will be conserved in the multiplication with $\Phi^{-1}$. \textcolor{redtmp}{On the other hand, $\Phi_w$ has the following structure.
\vspace{-3pt}
\begin{align*}%\label{eq_Phi_w_sparsity}
    \!\!\Phi_w = \!\!\lt[\matb
    0_{n \times n}\! & 0_{n \times n(T_s+1)} & \!0_{n \times n(T_f-1)}
    \\
    0_{n(T_s+1) \times n}\! & \Phi_{w,11} & \!0_{n(T_s+1) \times n(T_f-1)}
    \\
    0_{n(T_f-1) \times n}\! & \Phi_{w,21} & \!\Phi_{w,22}
    \mate\rt]\!\!, \!\!
\end{align*}
where $\Phi_{w,22}$ is lower-triangular to capture that future disturbances affect the prediction error in a causal way.} In the sequel, we will only write "$\Phi_v \mt{ causal}$" and "$\Phi_w \mt{ causal}$" to refer to these sparsity patterns.

In the next section, we explain how to handle the challenge that \eqref{eq_prob_def} includes an infinite-dimensional inf-sup problem.

%First, we want to reformulate \eqref{eq_prob_def_sls} to be able to apply it in real time. Second, we want to study how to guarantee asymptotic stability when the observer $\mc L$ is recomputed and applied at each $t$. Third, we will present an alternative imitation learning cost to reduce the conservativeness. Fourth and finally, we will discuss how to apply the same method for control (MPC).

\subsection{Risk closed-form solution}
\label{section_optim}

One of the most impactful results of DRO is its equivalence with regularization in regression problems \cite{wasserstein_kuhn}. The SLS reformulation allows us to use the DRO theory directly and express the risk in closed form.

\begin{theorem}\label{thm_regu}
The worst-case empirical risk \eqref{eq_cost_def_sls} can be written in closed form as %the sum of the empirical risk evaluated at past measurements of the noise and a regularization term. In mathematical terms, this gives
\begin{align}\label{eq_thm_prob}
    \mc R(\Phi_v, \Phi_w) = \lt\| \mc Q [\Phi_v, \Phi_w] \lt[\matb \underline{\tilde{\bs v}} \\ \underline{\tilde{\bs w}} \mate\rt] \rt\|_{F_1} + \|\mc Q[\varepsilon_v \Phi_v, \varepsilon_w \Phi_w]\|_{F_1},
\end{align}
where the $N$ empirical measurements are stacked as
\begin{subequations}%\label{eq_samples_def}
\begin{align*}
    [\tilde{\bs v}_1, \dots, \tilde{\bs v}_N] &= \underline{\tilde{\bs v}} \in \bb R^{\bs p \times N},
    \\
    [\tilde{\bs w}_1, \dots, \tilde{\bs w}_N] &= \underline{\tilde{\bs w}} \in \bb R^{\bs n \times N},
\end{align*}
\end{subequations}
\end{theorem}

\begin{proof}
Let $\kappa = \frac{\varepsilon_v}{\varepsilon_w}$ be the ratio between the Wasserstein balls radii, and let $\tilde{\bb P}_w'$ be the rescaled $\tilde{\bb P}_w$ distribution defined by
\vspace{-2pt}
\begin{align*}%\label{eq_proof_thm_regu_1}
    \tilde{\bb P}_w'(\kappa w) = \frac{1}{N}\sum_{i=1}^N \delta(\kappa(\bs w - \tilde{\bs w}_i)).
\end{align*}
Hence, with $\bs w' = \kappa \bs w$, we have
\begin{align*}%\label{eq_proof_thm_regu_2}
    \mc R(\Phi_v, \Phi_w) 
    =\!
    \sup_{\substack{
    \bb P_v \in \bb B_{\varepsilon_v}(\tilde{\bb P}_v)
    \\
    \bb P_w \in \bb B_{\varepsilon_v}(\tilde{\bb P}_w')}}
    \!
    \bb E_{\substack{v \sim \bb P_v \\ w' \sim \bb P_w}} 
    \!
    \|\mc Q(\Phi_v \bs v + \Phi_w \kappa^{-1} \bs w')\|_1.
\end{align*}

Theorem 10 in \cite{wasserstein_kuhn} states that with a Lipschitz cost $\ell(z) = \|\mc Q [\Phi_v, \kappa^{-1}\Phi_w] z\|_1$ and an unbounded support $\Xi$ in Definition \ref{def_wasser}, we have
\begin{align}
    \mc R(\Phi_v, \Phi_w) &=
    \bb E_{\substack{v \sim \tilde{\bb P}_v \\ w' \sim \tilde{\bb P}_w'}} 
    \|\mc Q(\Phi_v \bs v + \kappa^{-1} \Phi_w \bs w')\|_1
    \nonumber \\ \label{eq_proof_thm_regu_3}
    &\quad\quad + \varepsilon_v \sup_{z |\ell^\star(z) < +\infty} \|z\|_\star ,
\end{align}
where $\ell^\star$ is the convex conjugate function of $\ell$ and $\|\cdot\|_\star = \|\cdot\|_1$ is the dual to the $\ell_\infty$ norm used in Definition \ref{def_wasser}. The function $\ell^\star$ is given by
\begin{align*}%\label{eq_proof_thm_regu_4}
\ell^\star(z) &= \sup_{x \in \R^{\bs n + \bs p}} z^\top x - \ell(x),
\\ \nonumber
&= \sum_i^{\bs n + \bs p} \sup_{x_i} z_i x_i - |x_i| \lt\|\lt([\Phi_v, \kappa^{-1}\Phi_w]^\top \mc Q^\top\rt)_i\rt\|_1.
\end{align*}
Each of the supremums is either zero or infinite, depending on which of the two terms is larger in absolute value. This means that to obtain $\ell^\star(z) < +\infty$, each $|z_i|$ must not be greater than $\lt\|\lt([\Phi_v, \kappa^{-1}\Phi_w]^\top \mc Q^\top\rt)_i\rt\|_1$. Hence, the supremum in \eqref{eq_proof_thm_regu_3} is given by $\|\mc Q[\Phi_v, \kappa^{-1}\Phi_w]\|_{F1}$. To conclude the proof, we substitute this closed-form solution in \eqref{eq_proof_thm_regu_3} and compute explicitly the empirical expectation to obtain \eqref{eq_thm_prob}.
\end{proof}

Theorem \ref{thm_regu} gives a closed-form solution for the worst-case empirical risk $\mc R$. This removes the inner supremum in \eqref{eq_prob_def_sls}. \textcolor{redtmp}{Moreover, the resulting regularized cost is convex, which means that the infimum (11) is equal to a unique, global, and achievable minimum of the risk \eqref{eq_thm_prob}.}
%Moreover, the resulting regularized cost is convex, which means that \eqref{eq_prob_def_sls} is convex and has only one minimum that is global. 

\begin{corollary}\label{cor_split}
If $\mc Q$ is diagonal, then the problem \eqref{eq_prob_def_sls} can be split into $n$ separate optimization problems. The final solution of the full problem is given by
\begin{align}\label{eq_thm_ref}
    \Phi_{v} &= \mc Q^{-1} \lt[ \matb 0, \dots, 0, & 0, \dots, 0 
    \\
    (\argmin_{\phi} \|\Psi \phi - \mu_2\|_1)^\top, & 0, \dots, 0
    \\
    \vdots & \vdots
    \\ \vspace{-22pt} \\
    (\argmin_{\phi} \|\Psi \phi - \mu_{\bs n}\|_1)^\top, 
    & \rule{0pt}{20pt}\smash{\underbrace{0, \dots, 0}_{\bs p_0 \textnormal{ times}}} \mate \rt] \!,
    \\ \nonumber
    \\ \label{eq_thm_ref_proof_1}
    \Phi_w &= (I - \mc Z \mc A)^{-1} - \Phi_v \mc C \mc Z (I - \mc Z \mc A)^{-1},
\end{align}
where $\bs p_0 = p(T_f-1)$,
\begin{align*}%\label{eq_thm_ref_mu}
    \mu_i &= \lt(\lt[\mc Q (I - \mc Z \mc A)^{-1}\rt]_{\!i}\rt [ 0 \cdot \mc C^\top\!, \varepsilon_w I, \underline{\tilde{\bs w}}])^{\!\top} \quad \forall i = 2, \dots, \bs n,
\end{align*}
and $\Psi$ is the matrix formed by the $\bs p - \bs p_0$ first columns of
\begin{align*}%\label{eq_thm_ref_psi}
    \Psi_{nc} &= \lt[\matb - \varepsilon_v I \\ \varepsilon_w (I - \mc Z \mc A)^{-\top} \mc Z^\top \mc C^\top \\ \underline{\tilde{\bs w}}^\top (I - \mc Z \mc A)^{-\top} \mc Z^\top \mc C^\top - \underline{\tilde{\bs v}}^\top \mate\rt].
\end{align*}
\end{corollary}
\begin{proof}
First, we note that the constraint \eqref{eq_prob_achiev_def} is equivalent to \eqref{eq_thm_ref_proof_1}. Moreover, if $\Phi_v$ is causal, then so is $\Phi_w$ because both $(I - \mc Z \mc A)$ and $\mc C \mc Z$ are block lower-triangular. Plugging \eqref{eq_thm_ref_proof_1} into \eqref{eq_prob_def_sls} yields
\begin{align*}%\label{eq_thm_ref_proof_2}
    \argmin_{\Phi_v \mt{ causal}} &\lt\| \mc Q [\Phi_v, I_{\mc Z \mc A} - \Phi_v \mc C \mc Z I_{\mc Z \mc A}] \lt[\matb \underline{\tilde{\bs v}} \\ \underline{\tilde{\bs w}} \mate\rt] \rt\|_{F_1} 
    \\ \nonumber
    &+ \|\mc Q[\varepsilon_v \Phi_v, \varepsilon_w I_{\mc Z \mc A} - \varepsilon_w \Phi_v \mc C \mc Z I_{\mc Z \mc A}]\|_{F_1},
\end{align*}
where $I_{\mc Z \mc A} = (I - \mc Z \mc A)^{-1}$.
Rearranging the terms and inverting the sign inside the norm gives
\begin{align*}%\label{eq_thm_ref_proof_3}
    &\argmin_{\Phi_v \mt{ causal}} \lt\| \mc Q \Phi_v (\mc C \mc Z I_{\mc Z \mc A} \underline{\tilde{\bs w}} - \underline{\tilde{\bs v}}) - \mc Q I_{\mc Z \mc A} \underline{\tilde{\bs w}} \rt\|_{F_1} 
    \\ \nonumber
    &\quad + \|\mc Q \Phi_v [-\varepsilon_vI, \varepsilon_w\mc C \mc Z I_{\mc Z \mc A}] - \mc Q [0 \cdot \mc C^\top, \varepsilon_w I_{\mc Z \mc A}]\|_{F_1},
\end{align*}
where $0 \cdot \mc C^\top$ is used to construct a zero matrix of the same shape as $\Phi_v$.
By replacing the sum of norms by the norm of an augmented matrix, we obtain
\vspace{-3pt}
\begin{align}\label{eq_thm_ref_proof_4}
    \argmin_{\Phi_v \mt{ causal}} \| \mc Q &\; \Phi_v[-\varepsilon_v I, \varepsilon_w \mc C \mc Z I_{\mc Z \mc A}, \mc C \mc Z I_{\mc Z \mc A} \underline{\tilde{\bs w}} - \underline{\tilde{\bs v}}]
    \\ \nonumber
    &- \mc Q I_{\mc Z \mc A} [0 \cdot \mc C^\top, \varepsilon_w I, \underline{\tilde{\bs w}}] \|_{F_1}.
\end{align}
Let $\mc Q = \mt{diag}([q_1, \dots, q_{\bs n}])$, for $i = 1, \dots, \bs n$ the $i^{th}$ term of the Frobenius norm in \eqref{eq_thm_ref_proof_4} is written as
\begin{align}\label{eq_thm_ref_proof_5}
    \!\lt\| \! \lt[\matb -\varepsilon_v I \\ \varepsilon_w I_{\mc Z \mc A}^\top \mc Z^\top \mc C^\top \\ \underline{\tilde{\bs w}}^{\!\top} I_{\mc Z \mc A}^\top \mc Z^\top \mc C^\top \!\!-\! \underline{\tilde{\bs v}}^{\!\top} \! \mate\rt]
    \!\!
    (q_i \Phi_{v,i})^\top
    \!-\!
    \! \lt[\matb 0 \cdot \mc C \\ \varepsilon_w I \\ \underline{\tilde{\bs w}}^{\!\top} \mate\rt] \!\! (\mc Q I_{\mc Z \mc A})_i^\top \rt\|_{1} \!\!.\!
\end{align}
Note that each term only depends on the corresponding row of $\Phi_v$. Hence, the minimization problem can be separated into $\bs n$ independent sub-problems. Finally, one can solve for $\phi = q_i \Phi_{v,i}$ and plugging the constraint that $\Phi_v$ is causal simply removes the columns corresponding to the desired zeros in the matrix that pre-multiplies $(q_i \Phi_{v,i})^\top$ in \eqref{eq_thm_ref_proof_5}. Due to the sparisty pattern of $\mc Z^\top$ in \eqref{eq_thm_ref_proof_5}, the first $p$ entries of the optimizer $\phi$ are zero for all $i = 1, \dots, \bs n$, so we do not need to remove the first $p$ columns of $\Phi_v$.
\end{proof}

Theorem \ref{thm_regu} and Corollary \ref{cor_split} are the main results of this paper, as they allows to write \eqref{eq_prob_def_sls} as several small LPs. This allows one to first compute $\Phi_v$ row by row, and then obtain the observer policy $\mc L$ using
\begin{align*}%\label{eq_L_final}
    \mc L = \Phi_w^{-1} \Phi_v = (I - \mc Z \mc A)(I - \Phi_v \mc C \mc Z)^{-1} \Phi_v,
\end{align*}
which is implemented using \eqref{eq_observer_apply}.

\section{Numerical results}\label{section_results}
To highlight the ability of our method to handle time-varying and even nonlinear systems, we will perform experiments on a Van der Pol oscillator under complex disturbance patterns. The dynamics are given by
\vspace{-2pt}
\begin{align*}
    \dot x(t) &= \lt[x_2(t), \lt(1 - x_1(t)^2\rt)x_2(t) - x_1(t)\rt]^\top + w(t),
    \nonumber \\
    y &= x_1(t) + v(t).
\end{align*}
Note that this system is both continuous and non linear, so we cannot use \eqref{eq_thm_ref} directly. Hence, we discretize the system with a sampling frequency of 10Hz using the forward Euler method, and linearize it at each point of its trajectory (see Remark \ref{rem_nonlin} in Section \ref{subsec_dyn}) resulting in an LTV system. The time horizon considered is 1s (or 10 samples) and we are interested in the one step ahead prediction (i.e., $T_s = 8$ and $T_f = 1$). \textcolor{redtmp}{Additional details can be found in \cite{python_package}.}

\subsection{Noise profiles}
\vspace{-2pt}
We consider two different cases. First, we apply a noise profile following a sinusoidal pattern plus uniformly distributed noise of the same amplitude as the sine wave (see \cref{fig:sine_profile}). Second, we apply noise following a bi-modal noise distribution based on a mixture of two Gaussians. The precise distributions are:
\vspace{-4pt}
\begin{align}\label{eq_sin_noise}
    \!\!\![w_{t,s}^\top, v_{t,s}] \sim&\; \mc U(sin(10t)[0.1, 0.1, -0.1], 0.1),\!\!\!
    \\
    \!\!\![w_{t,b}^\top, v_{t,b}] \sim&\; 0.25\mc N(0.05[1, 1, -1], \mt{diag}(0.025[1,1,2]))\!\!\!
    \label{eq_bim_noise}
    \\ \nonumber
    &\!+ 0.75\mc N(0.05[-1, -1, 2], \mt{diag}(0.025[1, 1, 2])),
\end{align}
\vspace{-16pt} 

\noindent
where $\mc U(\mu, \sigma)$ is the uniform density supported by $[\mu-\sigma, \mu+\sigma]$ and $\mc N(\mu, \Sigma)$ is the Gaussian density with mean $\mu$ and variance $\Sigma$. In the sequel, we refer to \eqref{eq_sin_noise} and \eqref{eq_bim_noise} as the sine and bimodal noises, respectively. We generate 70 realizations at each time step and (i) use 20 of them as training data (shown in Fig. \ref{fig:noise_profiles}) to build $\underline{\tilde{\bs v}}$ and $\underline{\tilde{\bs w}}$ and (ii) the other 50 to validate the methods.
\vspace{-8pt} 
\begin{figure}[H]
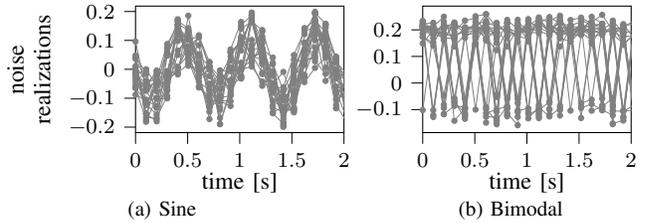

    %\centering
    \hspace{-10pt}
    \begin{subfigure}[b]{0.24\textwidth}
        \input{plots/sine_profile.tex}
        \vspace{-19pt}
        \subcaption{Sine}
        \vspace{-5pt}
        \label{fig:sine_profile}
    \end{subfigure}
    \hspace{1pt}
    \begin{subfigure}[b]{0.24\textwidth}
        \input{plots/bimodal_profile.tex}
        \vspace{-7pt}
        \subcaption{Bimodal}
        \vspace{-5pt}
        \label{fig:bimodal_profile}
    \end{subfigure}
    \hspace{-20pt}
    \caption{First two seconds of the 20 measurement noise realizations used to build the empirical distribution $\tilde{\bb P}_v$.}
    \label{fig:noise_profiles}
\end{figure}
\vspace{-14pt}
The non-stochastic components in \eqref{eq_sin_noise} and \eqref{eq_bim_noise} are a sine wave and a bias, respectively. %This leaves a uniform noise of magnitude 0.1 for the former, and a bimodal noise with a standard deviation up to 0.05 for the latter, as well as a linearization error for both, which has an $\ell_\infty$ norm measured on the training samples of less than 0.1 in both cases. 
Although the noise is bounded in \eqref{eq_sin_noise}, the linearization error may exceed the bound measured on the training samples. Hence, distributional robustness is well motivated for both types of noises.
\vspace{-4pt}

\vspace{-2pt}
\subsection{Results}\label{subsec_results}
\vspace{-2pt}
This section shows the prediction error (i.e. the last $n$ elements of $\bs e$) given by The EKF \cite{ekf_og}, unconstrained MHE with a quadratic cost \cite{book_mpc}, and distributionally robust MHE \eqref{eq_thm_ref} with  $\mc Q = \varepsilon I$ and $\varepsilon_w = \varepsilon_v = \varepsilon = 0.2$ (corresponding to the upper bound or the 95th percentile of the error), denoted in what follows by DRO\footnote{\textcolor{redtmp}{Solving MHE and DRO at each timestep takes about 0.04s and 0.1s, respectively, using Python on one core of a RaspberryPi.}}. The error is $\|\hat x_{t+1} - x_{t+1}\|_1$, where $\hat x_{t+1}$ is the prediction made by each method and $x_{t+1}$ is the exact state of the oscillator at time $t+1$.

Figures \ref{fig:sin} and \ref{fig:bim} show a significantly better estimation performance provided by distributionally robust MHE compared to classical MHE and the EKF. In particular, one can see in the error plot that the use of the empirical distribution mitigates the oscillations caused by the sine wave (Fig. \ref{fig:sin}) and the drift cause by the swings between the modes of the bimodal distribution (Fig. \ref{fig:bim}). Indeed, one can observe that in both cases, both MHE and EKF generate around 25\% to 60\% more error than our method.
\vspace{-4pt}

\vspace{-2pt}
\subsection{Wasserstein radius}
\vspace{-2pt}
In Section \ref{subsec_results}, we tuned the Wasserstein radius to be approximately equal to the sum of the magnitudes of the linearization error and the stochastic component in the noise. 

\noindent
Fig. \ref{fig:boxplot} also analyzes the performance of the naive data-driven estimator (i.e., when $\varepsilon = 0$) and shows that in both cases, while this approach still outperforms methods relying on the assumption of Gaussian noise, the mean and the variance of the error are much larger than with distributional robustness.% Indeed, the performance varies greatly depending on how close each test sample is to the training set, leading to very uncertain estimates.

\vspace{-5pt}
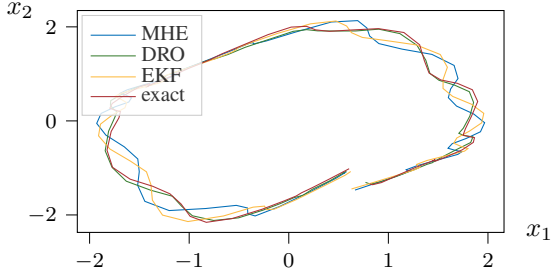
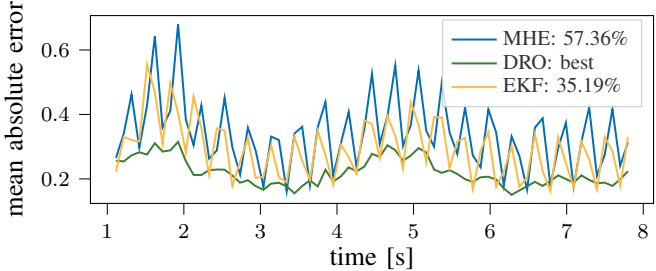
\begin{figure}[H]
    \begin{subfigure}{0.5\textwidth}
        %\centering
        \hspace{12pt}
        \input{plots/phase_sin}
        \vspace{-5pt}
        \subcaption{Phase diagram of the perturbed Van der Pol oscillator and its three state estimates.}
        \vspace{-1pt}
        \label{fig:phase_sin}
    \end{subfigure}
    \begin{subfigure}{0.5\textwidth}
        %\centering
        \input{plots/error_sin}
        \vspace{-5pt}
        \subcaption{Average error of three estimator at each time step and over all 50 test realizations. The legend shows the total relative error increments.}
        \vspace{-3pt}
        \label{fig:error_sin}
    \end{subfigure}
    \caption{Performance analysis of the EKF, MHE and distributionally robust MHE (DRO) under sine noise.}
    \label{fig:sin}
\end{figure}
\vspace{-22pt}
% \begin{figure}[H]
%     \begin{subfigure}{0.5\textwidth}
%         %\centering
%         \input{plots/phase_bim}
%         \vspace{-5pt}
%         \subcaption{Phase diagram of the perturbed Van der Pol oscillator and its three state estimates.}
%         \label{fig:phase_bim}
%     \end{subfigure}
%     \vspace{5pt}
%     \begin{subfigure}{0.5\textwidth}
%         %\centering
%         \input{plots/error_bim}
%         \vspace{-5pt}
%         \subcaption{Average error of three estimator at each time step and over all 50 test realizations. The legend shows the total relative error increases.}
%         \label{fig:error_bim}
%     \end{subfigure}
%     \caption{Performance analysis of the EKF, MHE and distributionally robust MHE (DRO) under bimodal noise.}
%     \label{fig:bim}
% \end{figure}
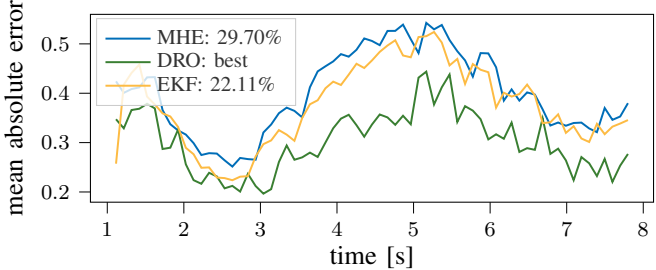
\begin{figure}[H]
    %\centering
    \input{plots/error_bim}
    \vspace{-18pt}
    \caption{Average error of the EKF, MHE and distributionally robust MHE (DRO) at each time step and over all 50 test realizations of bimodal noise. The legend shows the total relative error increments.}
    \label{fig:bim}
\end{figure}
\vspace{-19pt}
\begin{figure}[H]
    %\centering
    \input{plots/boxplot.tex}
    \vspace{-18pt}
    \caption{Statistics of the error of the EKF, MHE and distributionally robust MHE (DRO), and DRO with zero-radii Wasserstein balls over all 50 test realizations. The mean is shown with a dot marker. The total error is the sum of the errors at each time step.}
    \label{fig:boxplot}
\end{figure}
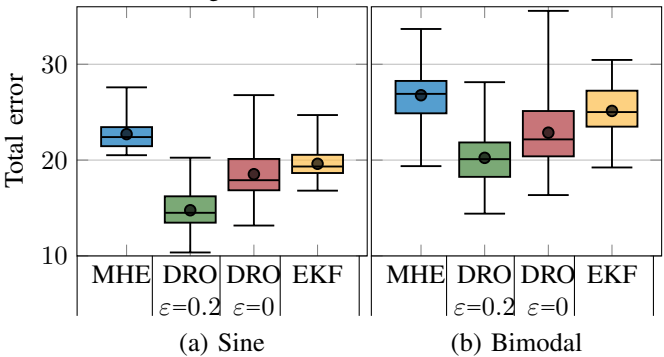
\vspace{-14pt}

\section{Conclusions}\label{section_conclu}
In this paper, we present a novel MHE method based on the empirical distribution of a system's noise and distributional 
robustness theory. \textcolor{redtmp}{We prove that our approach can be implemented as computationally-inexpensive LPs.} %\textcolor{redtmp}{Experiments on a challenging nonlinear system show high performance under complex noise profiles.}
%our approach can be implemented as a set of small LPs, which has a fairly low computational complexity. Moreover, we provide detailed numerical experiments on a challenging nonlinear system, and highlight the high performance of our method under complex noise profiles.

Future work will focus on studying quadratic costs, as they relate to energy or covariance. %in a system and are therefore more intuitive for engineers. 
Moreover, we will study how to include constraints in our formulation, as it is frequently done in MHE.

\vspace{-5pt}

\bibliographystyle{IEEEtran}
\bibliography{references}

\end{document}

%% file: plots/phase_sin.tex
% This file was created with tikzplotlib v0.9.15.
\begin{tikzpicture}

\begin{axis}[
width=0.8*\textwidth,
height=0.5\textwidth,
legend cell align={left},
legend style={
  fill opacity=0.8,
  draw opacity=1,
  text opacity=1,
  at={(0.01,0.98)},
  anchor=north west,
  draw=white!80!black,
  font=\footnotesize,
  row sep=-3pt
},
tick align=outside,
tick pos=left,
xlabel={$x_1$},
ylabel={$x_2$},
xlabel style={at={(1.03, -0.0)},anchor=west},
ylabel style={at={(-0.08, 1.0)},anchor=east,rotate=-90},
x grid style={white!69.0196078431373!black},
xmin=-2.12326941865828, xmax=2.16130737390566,
xtick style={color=black},
ticklabel style = {font=\footnotesize},
y grid style={white!69.0196078431373!black},
ymin=-2.371623974759, ymax=2.34744678700159,
ytick style={color=black}
]
\addplot [myb]
table {%
0.535836175846615 -1.10741085528546
0.563280411154501 -1.10331326275951
0.47161720350057 -1.22273552370436
0.227468965336247 -1.49394402659135
-0.0938083244803095 -1.8221825780984
-0.337796656246945 -2.02151748800621
-0.41080933985279 -1.97542941558759
-0.414840079095475 -1.85183695358332
-0.527012002863509 -1.79652646076079
-0.834822357922688 -1.86363568240096
-1.20130015219989 -1.90700129662218
-1.44520841936371 -1.7106355409269
-1.51741303109494 -1.37052728997364
-1.50129616022251 -1.08011273923944
-1.51613251561493 -0.809605659526821
-1.65128942406798 -0.551799502158075
-1.83669839551417 -0.302103573066463
-1.9285159280872 -0.0532209645002433
-1.88420714042404 0.159223770564817
-1.74412317990876 0.299905213369613
-1.60615206136104 0.401880007693954
-1.57931114490903 0.491347169971268
-1.63459783240816 0.54835384828605
-1.69452530574829 0.585085466542496
-1.62214192857774 0.661514234042312
-1.44596839602645 0.773833048427992
-1.22556722775452 0.898755704846394
-1.12283883653428 0.981753559799693
-1.14636651265978 0.997862025913044
-1.18295432953306 0.992240499982
-1.0883421106383 1.07336542065569
-0.867849711927004 1.24888880251332
-0.593482199134223 1.47969312212378
-0.40078694064601 1.64717174024099
-0.35854491856727 1.66799142838546
-0.356347903991221 1.6351563458914
-0.224022589480585 1.70115943827644
0.0614443704201122 1.88502045866048
0.427621203452358 2.10334302134738
0.691857423832062 2.13294357055793
0.783974509448214 1.97701968847732
0.812149840566909 1.80429437558652
0.904472223033231 1.64439187911078
1.13871782089114 1.52608164316415
1.43742191205378 1.41447900411286
1.66019211728233 1.17985891194755
1.69922892828063 0.900927141126108
1.62501098638636 0.679036126828481
1.60219774795327 0.493727685988758
1.68456237146869 0.315214020857707
1.85749629281822 0.145858500095664
1.96655388333457 -0.0388608522339274
1.92481869105944 -0.230998572555679
1.7814540035318 -0.373955979926399
1.64066151048754 -0.485240408489074
1.59909784121576 -0.581532605502942
1.66949292046444 -0.625095064894781
1.73665188749265 -0.655705392764502
1.69731756094016 -0.71212614240055
1.50583491985736 -0.824946612106761
1.29185611524643 -0.94987571484022
1.1788632834431 -1.03957648771847
1.1919782188899 -1.05898210638818
1.21486047896832 -1.05940135324172
1.15239314012383 -1.12036336299667
0.933645452197265 -1.28136896217471
0.664834864773077 -1.4710995169425
};
\addlegendentry{MHE}
\addplot [myg]
table {%
0.57945663222236 -1.06713490875819
0.414033119358229 -1.27211673975582
0.236064379850667 -1.45508636729636
0.142643366610413 -1.58224471905733
0.0523415350964571 -1.63255426357193
-0.0763207952690729 -1.75114262543203
-0.257455113860801 -1.89765625171569
-0.504478592470376 -2.06174125693415
-0.753852979965258 -2.11816643617252
-0.967669111795235 -2.01752837390628
-1.07751978333922 -1.81127751138425
-1.16862604389642 -1.60188185635855
-1.38703594350811 -1.47382631744709
-1.6312227844576 -1.28270034775071
-1.77457147558271 -0.987721544251481
-1.84349876059681 -0.634829109620251
-1.80921956443752 -0.218108786267981
-1.74911788180564 0.0645300036368141
-1.72489666019002 0.190562052174871
-1.77200595083294 0.274439599645463
-1.79068267934316 0.336383243921796
-1.78783083318796 0.430707712310907
-1.65773487127589 0.589006496191648
-1.52432519621279 0.715741533320236
-1.42057928949944 0.769986576164351
-1.40905977118888 0.75436844459889
-1.38890960900975 0.752701457827782
-1.30973914905155 0.838031975707526
-1.22140128155467 0.974459118250554
-1.03570875517325 1.10915677279029
-0.872936021498421 1.24356175619382
-0.804608531307239 1.29364521077754
-0.743319660203071 1.32065145802663
-0.612104449360544 1.41074805449557
-0.438393518401897 1.5642586932887
-0.211377478011398 1.72539516559225
0.0453342483968521 1.9052996963237
0.179716588206526 1.9407080816219
0.290389083167078 1.91481662953798
0.455622510407063 1.91152422276425
0.671324228972745 1.91251170258309
0.904301114516834 1.95079015470145
1.11639553739239 1.85616644667271
1.28755776467638 1.60752174884748
1.37600582074936 1.29016994323828
1.41732157822003 1.00911802449879
1.50840874980147 0.863633763827398
1.66648779480341 0.759055141796925
1.81313100287243 0.62083057594254
1.85579562876851 0.39259613170681
1.81796338804349 0.102021181451197
1.7523092387797 -0.155730331438625
1.70768414562016 -0.263530656686141
1.75984419077908 -0.312419630118323
1.8487145566776 -0.339216691946329
1.82678923591478 -0.459515223107307
1.68172297039863 -0.612383045203039
1.53393861649664 -0.774360198900061
1.47630135507274 -0.847741393739036
1.46457646781206 -0.858106102787155
1.43946262129758 -0.890674090141202
1.38684756293716 -0.924804299729456
1.25278108900289 -1.04287198021427
1.07361392426554 -1.1869419643965
0.920399840706837 -1.29678577664379
0.817866897875816 -1.34294531093378
0.770267792328397 -1.3075448397642
};
\addlegendentry{DRO}
\addplot [myy]
table {%
0.622000132883027 -1.07653051009412
0.565916338397868 -1.15585520213037
0.362041478419581 -1.37337922954411
0.0773267678836227 -1.67178073624806
-0.145802998544208 -1.87707440199047
-0.224767985883338 -1.88806096437999
-0.243732897886759 -1.81257606142065
-0.357825841687158 -1.83423809361981
-0.655576298023691 -2.02195708851374
-1.0125615062332 -2.14077998079664
-1.26772467330376 -2.00357714576256
-1.3693503340106 -1.68646233965439
-1.38413154378776 -1.35571979070426
-1.42885759767445 -1.08295752844151
-1.58903309235076 -0.856521050698988
-1.79592407799252 -0.598367898767491
-1.9102502264588 -0.317321396928417
-1.89043573785476 -0.0685076304153304
-1.77022713052773 0.12874484886008
-1.64487479235736 0.279523755139711
-1.62696364790605 0.388017224049119
-1.68690471618762 0.462175685434212
-1.74930964847073 0.516241984019009
-1.68420118580061 0.587156264112974
-1.52159964089622 0.684003276594911
-1.31574392218269 0.796878839439181
-1.22019348843277 0.878463757685939
-1.24192713429817 0.91623371022724
-1.27501624421201 0.935245737667208
-1.18817030315784 1.00425583033332
-0.987624559599589 1.14221805889684
-0.740330095553012 1.3202740022503
-0.563983657900565 1.4692466811575
-0.517911959962253 1.53328381726625
-0.506289646815768 1.56002829603338
-0.381029209803367 1.66107943339878
-0.118689322723076 1.85732317605399
0.218310668479722 2.06437206245157
0.473470325683136 2.12446233908395
0.580722869859924 2.01823056427454
0.629336801386008 1.87647458136565
0.737921766389843 1.77562809089983
0.976122169548912 1.70927691258172
1.2788391808433 1.61656067332061
1.52110545478371 1.41208623675215
1.59265032254331 1.14150029641954
1.54988891637295 0.885339686060701
1.54887222456378 0.689557072889638
1.64606158978486 0.520356005170952
1.83261104372375 0.35656021276921
1.95813687534299 0.161040506160164
1.93736769769934 -0.0409609636595394
1.81333023610122 -0.218494106210971
1.68742189252096 -0.363334386736482
1.65617530429402 -0.473947688694895
1.72911684941532 -0.541299832967269
1.79880042156738 -0.58851035302147
1.76499823591696 -0.642428893480047
1.58628927816655 -0.736439740715593
1.38753454695888 -0.84226905797159
1.28362216480843 -0.92523996022057
1.2956752104513 -0.96236949800078
1.31594341713075 -0.988269596909879
1.26054401659313 -1.04242713364741
1.06028598833199 -1.17386785125597
0.813867055894263 -1.33286264574642
0.629724473958303 -1.44770982623281
};
\addlegendentry{EKF}
\addplot [myr]
table {%
0.606148296120507 -1.01710366927895
0.406038986981657 -1.23815935069149
0.236662983870609 -1.4199647280089
0.128470162153924 -1.52884543211758
0.0730459730006812 -1.5986399120914
-0.0433944432812671 -1.69875381003785
-0.266271792262215 -1.87800148299253
-0.568481416915167 -2.07700965784753
-0.828212413972161 -2.15712075831534
-0.989958976244141 -2.03008938709211
-1.08083486910369 -1.768830275565
-1.19331208862334 -1.55056644283318
-1.36609680953194 -1.39070373425538
-1.59280337486757 -1.23950115127662
-1.76949973563851 -0.987388939240597
-1.8252264904882 -0.622745836625942
-1.78316350678334 -0.223806852446925
-1.70537563000316 0.0661536779259254
-1.68997373155101 0.198425410326763
-1.75428727192317 0.235090997513755
-1.79712209526251 0.269728247899298
-1.77292829221338 0.370561899656476
-1.64724027179585 0.551289897641517
-1.49962592459603 0.724478400636371
-1.39862255380368 0.790417797984819
-1.39700237146111 0.772911844726208
-1.41498645154152 0.737647718511971
-1.36525376431885 0.790768815296178
-1.20276902429203 0.930493725333649
-1.01395539995186 1.10956622680208
-0.861357700862109 1.24071970516084
-0.789957371983124 1.30117730273383
-0.759870633044858 1.33528335829801
-0.677016013255016 1.42102716543286
-0.464292203346108 1.6088351717925
-0.20724979169411 1.83546462919038
0.0237230429849691 1.99481836249439
0.177767747293414 2.01116088937134
0.270022426519753 1.935884847197
0.408330407497826 1.90482785500237
0.621483335377811 1.93245932577032
0.899601325702568 1.9623630303476
1.15555009659815 1.88013007999558
1.31424279261385 1.61831175045486
1.36109269665516 1.2829354894955
1.40812948953608 1.01982882486576
1.52652125258266 0.870879248507967
1.69193724315728 0.799676911239309
1.84353828072409 0.662370454452667
1.89709252921419 0.41205243033385
1.83342160046799 0.0951586479243466
1.77355805774502 -0.167505185882664
1.7504675099219 -0.303014974530675
1.804534754065 -0.330292095030011
1.86450665198957 -0.356636738575117
1.85299756915311 -0.457557331096009
1.74097799258069 -0.616142715424247
1.5738992967984 -0.770872441600298
1.47269854743636 -0.85567506888187
1.45376369869406 -0.846411275759885
1.47144210437694 -0.824597479332782
1.4245288481695 -0.867282034659268
1.27405915663868 -0.988534578806914
1.07746331039892 -1.17205354629694
0.930596086855932 -1.31178697731607
0.842044863503145 -1.35707703614368
0.813701645905613 -1.34153892128483
};
\addlegendentry{exact}
\end{axis}

\end{tikzpicture}

%% file: plots/error_sin.tex
% This file was created with tikzplotlib v0.9.15.
\begin{tikzpicture}

\begin{axis}[
width=\textwidth,
height=0.45\textwidth,
legend cell align={left},
legend style={fill opacity=0.8, draw opacity=1, text opacity=1, draw=white!80!black, font=\footnotesize, row sep=-2pt},
tick align=outside,
tick pos=left,
x grid style={white!69.0196078431373!black},
xlabel={time [s]},
xlabel style={at={(0.5, -0.4)},anchor=south},
xmin=0.78, xmax=8.13428571428572,
xtick style={color=black},
ticklabel style = {font=\footnotesize},
y grid style={white!69.0196078431373!black},
ylabel={mean absolute error},
ymin=0.123922017228964, ymax=0.706684252099151,
ytick style={color=black}
]
\addplot [line width=.75pt, myb]
table {%
1.11428571428571 0.265268694848518
1.21558441558442 0.341306875124816
1.31688311688312 0.461548114365519
1.41818181818182 0.297057957671187
1.51948051948052 0.427437997242608
1.62077922077922 0.64236256473966
1.72207792207792 0.359341022379236
1.82337662337662 0.41123959536866
1.92467532467532 0.680195059605051
2.02597402597403 0.383870445481183
2.12727272727273 0.305012744049434
2.22857142857143 0.427639946364165
2.32987012987013 0.263037312303619
2.43116883116883 0.288442411562715
2.53246753246753 0.451030592831328
2.63376623376623 0.297797927701096
2.73506493506494 0.21311294641446
2.83636363636364 0.358716068828747
2.93766233766234 0.287855444616664
3.03896103896104 0.176020950458192
3.14025974025974 0.331147841954677
3.24155844155844 0.320588878371658
3.34285714285714 0.161897467921101
3.44415584415584 0.340957758526903
3.54545454545455 0.361649099486549
3.64675324675325 0.179952079867792
3.74805194805195 0.360016135874656
3.84935064935065 0.438153245047996
3.95064935064935 0.209253378187411
4.05194805194805 0.311518559089348
4.15324675324675 0.407987900642309
4.25454545454545 0.23943699327669
4.35584415584416 0.352308193305289
4.45714285714286 0.525986513936242
4.55844155844156 0.308953778514881
4.65974025974026 0.395668107610498
4.76103896103896 0.551331507406188
4.86233766233766 0.302795334090518
4.96363636363636 0.366572173530654
5.06493506493507 0.53564451098584
5.16623376623377 0.348852633322538
5.26753246753247 0.300291016792692
5.36883116883117 0.497769601169619
5.47012987012987 0.344433190255625
5.57142857142857 0.277359779074848
5.67272727272727 0.420803387794573
5.77402597402597 0.272460120882733
5.87532467532468 0.236117988827627
5.97662337662338 0.415251033351141
6.07792207792208 0.344822668794131
6.17922077922078 0.179530112775107
6.28051948051948 0.331524586695947
6.38181818181818 0.271721369865066
6.48311688311688 0.166488315143705
6.58441558441559 0.358639026559502
6.68571428571429 0.387992172296871
6.78701298701299 0.179578416142641
6.88831168831169 0.295149575632062
6.98961038961039 0.374185257334519
7.09090909090909 0.175674629328317
7.19220779220779 0.315358980388646
7.29350649350649 0.42172593652835
7.39480519480519 0.210072696961654
7.4961038961039 0.27423857944075
7.5974025974026 0.416087760689862
7.6987012987013 0.24138336164234
7.8 0.313025698964076
};
\addlegendentry{MHE: 57.36\%}
\addplot [line width=.75pt, myg]
table {%
1.11428571428571 0.256707668286248
1.21558441558442 0.254201851849239
1.31688311688312 0.273187867049613
1.41818181818182 0.282550764240813
1.51948051948052 0.275649524539106
1.62077922077922 0.310679094631184
1.72207792207792 0.284795709028553
1.82337662337662 0.288626094647489
1.92467532467532 0.314774816472866
2.02597402597403 0.255242471047089
2.12727272727273 0.212265414316833
2.22857142857143 0.211980010809762
2.32987012987013 0.227057900648765
2.43116883116883 0.229081998608659
2.53246753246753 0.228642776923719
2.63376623376623 0.210391667539615
2.73506493506494 0.188232843654424
2.83636363636364 0.196018602257151
2.93766233766234 0.177728130114444
3.03896103896104 0.166200855431531
3.14025974025974 0.18533439813191
3.24155844155844 0.188015234064215
3.34285714285714 0.176587517334452
3.44415584415584 0.155309344020304
3.54545454545455 0.177803447297431
3.64675324675325 0.195901790063422
3.74805194805195 0.176434965837777
3.84935064935065 0.228371697113778
3.95064935064935 0.191393054308845
4.05194805194805 0.208183524031545
4.15324675324675 0.235976021035489
4.25454545454545 0.223040753051439
4.35584415584416 0.237945843208769
4.45714285714286 0.277616148713874
4.55844155844156 0.268293197866794
4.65974025974026 0.304436256125572
4.76103896103896 0.289139315071614
4.86233766233766 0.256110312709414
4.96363636363636 0.272404600198641
5.06493506493507 0.295579493748541
5.16623376623377 0.282107107917356
5.26753246753247 0.229201467358193
5.36883116883117 0.218724450237995
5.47012987012987 0.226941571395258
5.57142857142857 0.215208042060497
5.67272727272727 0.19906562055598
5.77402597402597 0.190413939611913
5.87532467532468 0.205137017573127
5.97662337662338 0.206669049174743
6.07792207792208 0.198030962506156
6.17922077922078 0.170317859792212
6.28051948051948 0.150411209723063
6.38181818181818 0.164190000883283
6.48311688311688 0.177368421900929
6.58441558441559 0.190950030236798
6.68571428571429 0.178271195022938
6.78701298701299 0.192902029731604
6.88831168831169 0.210984874311321
6.98961038961039 0.199716276654605
7.09090909090909 0.188294235864902
7.19220779220779 0.211053852950973
7.29350649350649 0.196777526234945
7.39480519480519 0.186542669771175
7.4961038961039 0.18843263456008
7.5974025974026 0.178095757974004
7.6987012987013 0.198587353673633
7.8 0.223549810927477
};
\addlegendentry{DRO: best}
\addplot [line width=.75pt, myy]
table {%
1.11428571428571 0.222463498946476
1.21558441558442 0.330279404365533
1.31688311688312 0.320603628469022
1.41818181818182 0.316496425373028
1.51948051948052 0.554624501005549
1.62077922077922 0.469998111737742
1.72207792207792 0.308633848940825
1.82337662337662 0.498048097749156
1.92467532467532 0.404083922659114
2.02597402597403 0.281259021834904
2.12727272727273 0.455492082686344
2.22857142857143 0.35820220734959
2.32987012987013 0.213129742880239
2.43116883116883 0.356793429228946
2.53246753246753 0.35004191877988
2.63376623376623 0.178594158129134
2.73506493506494 0.25428881179411
2.83636363636364 0.326692000745021
2.93766233766234 0.203044186886405
3.03896103896104 0.206639329943536
3.14025974025974 0.304882120491235
3.24155844155844 0.204518517835576
3.34285714285714 0.189766048193396
3.44415584415584 0.331297297434867
3.54545454545455 0.256330967218687
3.64675324675325 0.196958267857643
3.74805194805195 0.349655603648456
3.84935064935065 0.274198096932908
3.95064935064935 0.182063698082575
4.05194805194805 0.306477412986262
4.15324675324675 0.267339580172133
4.25454545454545 0.215753478144295
4.35584415584416 0.380786797988098
4.45714285714286 0.369933117432978
4.55844155844156 0.264530969931512
4.65974025974026 0.395320526048709
4.76103896103896 0.333873314224302
4.86233766233766 0.232884615132451
4.96363636363636 0.437790759212584
5.06493506493507 0.368628367497175
5.16623376623377 0.251176709622911
5.26753246753247 0.391829264706011
5.36883116883117 0.390999583745883
5.47012987012987 0.234158124289705
5.57142857142857 0.316058198193437
5.67272727272727 0.326865449508327
5.77402597402597 0.167927413662782
5.87532467532468 0.286822961775003
5.97662337662338 0.346029968707311
6.07792207792208 0.17552967765928
6.17922077922078 0.223211634602046
6.28051948051948 0.302816180506147
6.38181818181818 0.175913963330991
6.48311688311688 0.204463919706857
6.58441558441559 0.338869603098474
6.68571428571429 0.224145596921748
6.78701298701299 0.162146690858782
6.88831168831169 0.329716397334951
6.98961038961039 0.248731299752999
7.09090909090909 0.174503987727748
7.19220779220779 0.31642808354138
7.29350649350649 0.260691095474502
7.39480519480519 0.167418115094799
7.4961038961039 0.330102847750827
7.5974025974026 0.286207473758651
7.6987012987013 0.177867584131665
7.8 0.330582999148076
};
\addlegendentry{EKF: 35.19\%}
\end{axis}

\end{tikzpicture}

%% file: plots/error_bim.tex
% This file was created with tikzplotlib v0.9.15.
\begin{tikzpicture}

\begin{axis}[
width=0.5\textwidth,
height=0.225\textwidth,
legend cell align={left},
legend style={fill opacity=0.8, draw opacity=1, text opacity=1, draw=white!80!black, at={(0.01,0.98)},anchor=north west, font=\footnotesize, row sep=-2pt},
tick align=outside,
tick pos=left,
x grid style={white!69.0196078431373!black},
xlabel={time [s]},
xlabel style={at={(0.5, -0.4)},anchor=south},
xmin=0.78, xmax=8.13428571428572,
xtick style={color=black},
ticklabel style = {font=\footnotesize},
y grid style={white!69.0196078431373!black},
ylabel={mean absolute error},
ymin=0.179297418249312, ymax=0.559789547160593,
ytick style={color=black}
]
\addplot [line width=.75pt, myb]
table {%
1.11428571428571 0.424063570572292
1.21558441558442 0.400778828557207
1.31688311688312 0.408508220582805
1.41818181818182 0.411857665744449
1.51948051948052 0.432171405151184
1.62077922077922 0.432677800303994
1.72207792207792 0.36597545179431
1.82337662337662 0.337220026285853
1.92467532467532 0.324651260107902
2.02597402597403 0.316315988348369
2.12727272727273 0.298925874160548
2.22857142857143 0.274860043261464
2.32987012987013 0.278594322765839
2.43116883116883 0.277782913395554
2.53246753246753 0.264096105307053
2.63376623376623 0.251469749245122
2.73506493506494 0.268902952287643
2.83636363636364 0.266574404866785
2.93766233766234 0.265364327658704
3.03896103896104 0.32050129927704
3.14025974025974 0.336677130690189
3.24155844155844 0.359554931854449
3.34285714285714 0.37081091783005
3.44415584415584 0.364266620563432
3.54545454545455 0.351139645674097
3.64675324675325 0.412137017186103
3.74805194805195 0.444250645640918
3.84935064935065 0.45458596276145
3.95064935064935 0.464459823236918
4.05194805194805 0.479240867414253
4.15324675324675 0.47407028708235
4.25454545454545 0.487988799027958
4.35584415584416 0.511008065968161
4.45714285714286 0.505874920834335
4.55844155844156 0.497360481712782
4.65974025974026 0.526230268590986
4.76103896103896 0.526317576080992
4.86233766233766 0.53898287539474
4.96363636363636 0.509465329976225
5.06493506493507 0.481437945484186
5.16623376623377 0.542494450391899
5.26753246753247 0.529453223525476
5.36883116883117 0.538281727114125
5.47012987012987 0.507753542007963
5.57142857142857 0.485147896514447
5.67272727272727 0.466525786527855
5.77402597402597 0.433604167014214
5.87532467532468 0.481325760288011
5.97662337662338 0.480729392672398
6.07792207792208 0.452597781726513
6.17922077922078 0.385085317489399
6.28051948051948 0.407772847901622
6.38181818181818 0.385089765353413
6.48311688311688 0.401680601547317
6.58441558441559 0.397265474713424
6.68571428571429 0.367356291512983
6.78701298701299 0.33459734580245
6.88831168831169 0.340691274219943
6.98961038961039 0.334158848206799
7.09090909090909 0.339727663817109
7.19220779220779 0.340516701122317
7.29350649350649 0.328698033554182
7.39480519480519 0.321137266347446
7.4961038961039 0.370099182529463
7.5974025974026 0.346105788274403
7.6987012987013 0.352867192049422
7.8 0.37986449862744
};
\addlegendentry{MHE: 29.70\%}
\addplot [line width=.75pt, myg]
table {%
1.11428571428571 0.347300800943423
1.21558441558442 0.328915556887792
1.31688311688312 0.366015407704966
1.41818181818182 0.368582714551359
1.51948051948052 0.378681296308817
1.62077922077922 0.369322289664185
1.72207792207792 0.287044209251768
1.82337662337662 0.289707948472429
1.92467532467532 0.326396746638199
2.02597402597403 0.256065215507307
2.12727272727273 0.22392642432291
2.22857142857143 0.216591518360134
2.32987012987013 0.238820082463439
2.43116883116883 0.231489602423208
2.53246753246753 0.207383779075758
2.63376623376623 0.212326581758247
2.73506493506494 0.200814437782423
2.83636363636364 0.23668991221765
2.93766233766234 0.212022296973028
3.03896103896104 0.196592515018007
3.14025974025974 0.205701247454482
3.24155844155844 0.260660109083997
3.34285714285714 0.294173831478185
3.44415584415584 0.265261108731792
3.54545454545455 0.270079180421451
3.64675324675325 0.279719734643596
3.74805194805195 0.270911996247968
3.84935064935065 0.300151499632677
3.95064935064935 0.329072437166216
4.05194805194805 0.348573989839857
4.15324675324675 0.356356493718967
4.25454545454545 0.335351877681906
4.35584415584416 0.311981530304362
4.45714285714286 0.336667498010165
4.55844155844156 0.356968130177232
4.65974025974026 0.351044006485634
4.76103896103896 0.336050721744162
4.86233766233766 0.354405045439331
4.96363636363636 0.337146174191777
5.06493506493507 0.431565040173416
5.16623376623377 0.443759075308328
5.26753246753247 0.3772249624697
5.36883116883117 0.412054563282109
5.47012987012987 0.437919366380543
5.57142857142857 0.341816811422231
5.67272727272727 0.373829440898681
5.77402597402597 0.364796561624305
5.87532467532468 0.347505787087178
5.97662337662338 0.306521769642818
6.07792207792208 0.316857799264279
6.17922077922078 0.307785788603498
6.28051948051948 0.340392218328154
6.38181818181818 0.281052525255179
6.48311688311688 0.288737620774465
6.58441558441559 0.285370343425921
6.68571428571429 0.34916149246159
6.78701298701299 0.277144875203656
6.88831168831169 0.287706445110112
6.98961038961039 0.263978225883318
7.09090909090909 0.223895822675042
7.19220779220779 0.270682762837604
7.29350649350649 0.258146817010134
7.39480519480519 0.232209181097806
7.4961038961039 0.266919157576492
7.5974025974026 0.220479898860603
7.6987012987013 0.253397307591238
7.8 0.276975798684231
};
\addlegendentry{DRO: best}
\addplot [line width=.75pt, myy]
table {%
1.11428571428571 0.257319291033194
1.21558441558442 0.417338062621449
1.31688311688312 0.441137862737068
1.41818181818182 0.458048231693034
1.51948051948052 0.393235639327496
1.62077922077922 0.377216288128391
1.72207792207792 0.35922492586448
1.82337662337662 0.352938896800387
1.92467532467532 0.33166300927116
2.02597402597403 0.289448935564799
2.12727272727273 0.27646079594635
2.22857142857143 0.24899675427752
2.32987012987013 0.24978689141935
2.43116883116883 0.229220728547733
2.53246753246753 0.228088808699698
2.63376623376623 0.223794866942804
2.73506493506494 0.230974099583461
2.83636363636364 0.232022353089421
2.93766233766234 0.272018260627017
3.03896103896104 0.296888764433571
3.14025974025974 0.304445647266234
3.24155844155844 0.325319758208337
3.34285714285714 0.316165073299499
3.44415584415584 0.303663121752455
3.54545454545455 0.349942180856736
3.64675324675325 0.377616849325458
3.74805194805195 0.385793516040279
3.84935064935065 0.410077633462462
3.95064935064935 0.423475489718726
4.05194805194805 0.416719459934372
4.15324675324675 0.437687004431162
4.25454545454545 0.459677892525915
4.35584415584416 0.451122956586334
4.45714285714286 0.466479537937771
4.55844155844156 0.483816928750852
4.65974025974026 0.496682911593701
4.76103896103896 0.507264468215188
4.86233766233766 0.476452420613914
4.96363636363636 0.472849810967117
5.06493506493507 0.513841770002825
5.16623376623377 0.51559038545208
5.26753246753247 0.524122082768197
5.36883116883117 0.502654488721997
5.47012987012987 0.456427036164855
5.57142857142857 0.469479289959978
5.67272727272727 0.418893620286765
5.77402597402597 0.459116040613238
5.87532467532468 0.447594453957518
5.97662337662338 0.442562912146245
6.07792207792208 0.370843517669732
6.17922077922078 0.401050008487791
6.28051948051948 0.393278038209485
6.38181818181818 0.39835599418263
6.48311688311688 0.417478036757002
6.58441558441559 0.395617284777455
6.68571428571429 0.339343730233193
6.78701298701299 0.34121561658331
6.88831168831169 0.357219050614968
6.98961038961039 0.319736864764382
7.09090909090909 0.333178408834365
7.19220779220779 0.308300713913134
7.29350649350649 0.301283875135864
7.39480519480519 0.337998588380171
7.4961038961039 0.317032413193636
7.5974025974026 0.332286915415887
7.6987012987013 0.338639932796534
7.8 0.345634586876935
};
\addlegendentry{EKF: 22.11\%}
\end{axis}

\end{tikzpicture}

%% file: plots/boxplot.tex
\makeatletter
\pgfplotsset{
    boxplot prepared from table/.code={
        \def\tikz@plot@handler{\pgfplotsplothandlerboxplotprepared}%
        \pgfplotsset{
            /pgfplots/boxplot prepared from table/.cd,
            #1,
        }
    },
    /pgfplots/boxplot prepared from table/.cd,
        table/.code={\pgfplotstablecopy{#1}\to\boxplot@datatable},
        row/.initial=0,
        make style readable from table/.style={
            #1/.code={
                \pgfplotstablegetelem{\pgfkeysvalueof{/pgfplots/boxplot prepared from table/row}}{##1}\of\boxplot@datatable
                \pgfplotsset{boxplot/#1/.expand once={\pgfplotsretval}}
            }
        },
        make style readable from table=lower whisker,
        make style readable from table=upper whisker,
        make style readable from table=lower quartile,
        make style readable from table=upper quartile,
        make style readable from table=median,
        make style readable from table=lower notch,
        make style readable from table=upper notch
}
\makeatother

\pgfplotstableread{
    lw    lq    med   uq    uw
    20.51 21.45 22.40 23.44 27.59
    10.36 13.47 14.50 16.22 20.25
    13.17 16.85 17.90 20.12 26.78
    16.81 18.65 19.33 20.55 24.70
}\sindatatable

\pgfplotstableread{
    lw    lq    med   uq    uw
    19.37 24.88 26.91 28.26 33.68
    14.41 18.25 20.10 21.84 28.13
    16.35 20.39 22.16 25.12 35.57
    19.23 23.49 25.01 27.24 30.45
}\bimdatatable

\begin{tikzpicture} 
\begin{axis}[
width=0.3\textwidth,
height=0.27\textwidth,
xtick={1,2,3,4},
xticklabel style={align=center},
xticklabels={
\hspace{-4pt}MHE,
{DRO\\$\varepsilon \texttt{=} 0.2$},
{DRO\\$\varepsilon \texttt{=} 0$},
EKF},
ylabel={Total error},
xlabel={(a) Sine},
ymin=10,
ymax=36,
ymajorgrids=true,
cycle list={{myb},{myg},{myr},{myy}},
boxplot/draw direction=y,
]
  \addplot+[mark = *, mark options = {black},
  fill, fill opacity=0.67, draw=black,
  line width=0.25mm,
  boxplot prepared from table={
    table=\sindatatable,
    lower whisker=lw,
    upper whisker=uw,
    lower quartile=lq,
    upper quartile=uq,
    median=med
  }, boxplot prepared
  ]
  coordinates {(0, 22.710652026239)};
  \addplot+[mark = *, mark options = {black},
  fill, fill opacity=0.67, draw=black,
  line width=0.25mm,
  boxplot prepared from table={
    table=\sindatatable,
    row=1,
    lower whisker=lw,
    upper whisker=uw,
    lower quartile=lq,
    upper quartile=uq,
    median=med
  }, boxplot prepared
  ]
  coordinates {(1, 14.765837944636091)};
  \addplot+[mark = *, mark options = {black},
  fill, fill opacity=0.67, draw=black,
  line width=0.25mm,
  boxplot prepared from table={
    table=\sindatatable,
    row=2,
    lower whisker=lw,
    upper whisker=uw,
    lower quartile=lq,
    upper quartile=uq,
    median=med
  }, boxplot prepared
  ]
  coordinates {(2, 18.52822745649203)};
  \addplot+[mark = *, mark options = {black},
  fill, fill opacity=0.67, draw=black,
  line width=0.25mm,
  boxplot prepared from table={
    table=\sindatatable,
    row=3,
    lower whisker=lw,
    upper whisker=uw,
    lower quartile=lq,
    upper quartile=uq,
    median=med
  }, boxplot prepared
  ]
  coordinates {(3, 19.613612710581695)};
\end{axis}
\draw(0.0,-0.8) -- (0.0,0);
\draw(1.0,-0.8) -- (1.0,0);
\draw(1.9,-0.8) -- (1.9,0);
\draw(2.7,-0.8) -- (2.7,0);
\draw(3.75,-0.8) -- (3.75,0);
\end{tikzpicture}
\hspace{-13pt}
\begin{tikzpicture} 
\begin{axis}[
width=0.3\textwidth,
height=0.27\textwidth,
xtick={1,2,3,4},
xticklabel style={align=center},
xticklabels={
\hspace{-4pt}MHE,
{DRO\\$\varepsilon \texttt{=} 0.2$},
{DRO\\$\varepsilon \texttt{=} 0$},
EKF},
xlabel={(b) Bimodal},
ymin=10,
ymax=36,
yticklabels={,,},
ymajorgrids=true,
cycle list={{myb},{myg},{myr},{myy}},
boxplot/draw direction=y,
]
  \addplot+[mark = *, mark options = {black},
  fill, fill opacity=0.67, draw=black,
  line width=0.25mm,
  boxplot prepared from table={
    table=\bimdatatable,
    lower whisker=lw,
    upper whisker=uw,
    lower quartile=lq,
    upper quartile=uq,
    median=med
  }, boxplot prepared
  ]
  coordinates {(0, 26.76778014353073)};
  \addplot+[mark = *, mark options = {black},
  fill, fill opacity=0.67, draw=black,
  line width=0.25mm,
  boxplot prepared from table={
    table=\bimdatatable,
    row=1,
    lower whisker=lw,
    upper whisker=uw,
    lower quartile=lq,
    upper quartile=uq,
    median=med
  }, boxplot prepared
  ]
  coordinates {(1, 20.23285340771144)};
  \addplot+[mark = *, mark options = {black},
  fill, fill opacity=0.67, draw=black,
  line width=0.25mm,
  boxplot prepared from table={
    table=\bimdatatable,
    row=2,
    lower whisker=lw,
    upper whisker=uw,
    lower quartile=lq,
    upper quartile=uq,
    median=med
  }, boxplot prepared
  ]
  coordinates {(2, 22.860843300135098)};
  \addplot+[mark = *, mark options = {black},
  fill, fill opacity=0.67, draw=black,
  line width=0.25mm,
  boxplot prepared from table={
    table=\bimdatatable,
    row=3,
    lower whisker=lw,
    upper whisker=uw,
    lower quartile=lq,
    upper quartile=uq,
    median=med
  }, boxplot prepared
  ]
  coordinates {(3, 25.12985178499507)};
\end{axis}
\draw(0.0,-0.8) -- (0.0,0);
\draw(1.0,-0.8) -- (1.0,0);
\draw(1.9,-0.8) -- (1.9,0);
\draw(2.7,-0.8) -- (2.7,0);
\draw(3.75,-0.8) -- (3.75,0);
\end{tikzpicture}